\documentclass[a4paper,UKenglish,cleveref,autoref,thm-restate]{lipics-v2019}

\usepackage[utf8]{inputenc}
\usepackage[T1]{fontenc}

\usepackage{hyperref}
\usepackage{enumerate}

\usepackage{boxedminipage}

\allowdisplaybreaks

\usepackage{stmaryrd}
\usepackage{amsthm,amssymb,amsmath}
\usepackage{mathtools}
\usepackage{ebproof}
\usepackage{braket}
\usepackage{proof}
\newcommand{\rinfer}[2]{\infer{#2}{#1}}

\usepackage{xcolor}
\usepackage{graphicx}
\usepackage{tikz}
\usetikzlibrary{matrix}
\usepackage{tikz-cd}

\theoremstyle{plain}
\newtheorem{conjecture} [theorem]{Conjecture}
\newtheorem*{conjecture*} {Conjecture}
\theoremstyle{definition}
\newtheorem{question}[theorem]{Question}
\newtheorem*{question*}{Question}

\def\emph{\textbf}

\newcommand{\ie}{i.e.~}
\newcommand{\eg}{e.g.~}

\newcommand{\Mcomma}{\text{ ,}}

\newcommand{\Mdot}{\text{ .}}

\newcommand{\Miff}{\quad\quad\text{ iff }\quad\quad}

\usepackage{colonequals}
\newcommand{\defeq}{\colonequals} 
\newcommand{\tuple}[1]{\mathopen{\langle}#1\mathclose{\rangle}} 
\newcommand{\setdef}  [2]{\left\{#1 \mid #2\right\}}             
\newcommand{\enset}   [1]{\mathopen{ \{ }#1\mathclose{ \} }} 
\newcommand{\family}   [1]{\mathopen{ ( }#1\mathclose{ ) }} 
\newcommand{\fdec}    [3]{#1 \colon #2 \longrightarrow #3}

\renewcommand{\to}{\longrightarrow}

\newcommand{\pfn}{\relbar\joinrel\rightharpoonup}

\newcommand{\es}{\varnothing}

\newcommand{\join}{\vee}
\newcommand{\meet}{\wedge}
\newcommand{\bigjoin}{\bigvee}

\newcommand{\defIFF}{\; \ratio\Leftrightarrow \;}
\newcommand{\IMP}{\; \Rightarrow \;}
\newcommand{\AND}{\; \wedge \;}

\newcommand{\ocirc}{\circledcirc}

\newcommand{\CC}{\mathbb{C}}

\newcommand{\cat}[1]{\ensuremath{\mathbf{#1}}}
\newcommand{\ba}{\cat{BA}}
\newcommand{\pba}{\cat{pBA}}
\newcommand{\epba}{\cat{epBA}}
\newcommand{\Hilb}{\cat{Hilb}}

\newcommand{\XX}{\mathbf{X}}
\newcommand{\XGO}{\tuple{X,\adj,O}}
\newcommand{\adj}{{\frown}}
\DeclareMathOperator{\Clique}{\mathsf{Kl}}
\DeclareMathOperator{\Ev}{\mathcal{E}}
\newcommand{\vecs}{\mathbf{s}}
\newcommand{\vect}{\mathbf{t}}

\newcommand{\comm}{\odot}

\DeclareMathOperator{\C}{\mathcal{C}}
\newcommand{\CA}{\C(A)}

\newcommand{\lex}{\perp}
\newcommand{\Alex}{{A[\lex]}}
\newcommand{\Alexstar}{{A[\lex]^*}}
\newcommand{\Exc}{X}

\newcommand{\inl}{\imath}
\newcommand{\inr}{\jmath}
\newcommand{\ex}{{\downarrow}}
\newcommand{\eqBool}{\equiv_{\mathsf{Bool}}}

\newcommand{\AX}{A_{\XX}}
\newcommand{\BX}{B_{\XX}}

\DeclareMathOperator{\Proj}{\mathsf{P}}
\newcommand{\HH}{\mathcal{H}}
\newcommand{\KK}{\mathcal{K}}
\newcommand{\PH}{\Proj(\HH)}
\newcommand{\PK}{\Proj(\KK)}
\newcommand{\PHoK}{\Proj(\HH\otimes\KK)}

\newcommand{\vphi}{\varphi}

\newcommand{\One}{\mathbf{1}}
\newcommand{\Two}{\mathbf{2}}
\newcommand{\Four}{\mathbf{4}}

\title{The logic of contextuality} 
\titlerunning{The logic of contextuality}

\author{\href{https://www.cs.ox.ac.uk/samson.abramsky}{Samson Abramsky}}{Department of Computer Science, University of Oxford,
United Kingdom
}{samson.abramsky@cs.ox.ac.uk}{https://orcid.org/0000-0003-3921-6637}{This author acknowledges support from EPSRC -- Engineering and Physical Sciences Research Council, EP/T00696X/1, \textit{Resources and Coresources: a junction between categorical semantics and descriptive complexity}.}

\author{\href{https://www.cs.ox.ac.uk/people/rui.soaresbarbosa/rsb}{Rui Soares Barbosa}\footnote{This work was carried out in part while RSB was at the School of Informatics, the University of Edinburgh.}}
{INL -- International Iberian Nanotechnology Laboratory, 
Portugal
}{rui.soaresbarbosa@inl.int}{https://orcid.org/0000-0002-0465-8518}{This author acknowledges support from EPSRC -- Engineering and Physical Sciences Research Council, EP/R044759/1, \textit{Combining Viewpoints in Quantum Theory (Ext.)}, and from FCT -- Funda{\c{c}\~ao} para a Ci{\^e}ncia e a Tecnologia, CEECINST/00062/2018.}

\authorrunning{S. Abramsky and R.\,S. Barbosa}

\Copyright{Samson Abramsky and Rui Soares Barbosa}



\acknowledgements{The authors would like to thank Chris Heunen for helpful discussions.}

\nolinenumbers 

\EventEditors{Christel Baier and Jean Goubault-Larrecq}
\EventNoEds{2}
\EventLongTitle{29th EACSL Annual Conference on Computer Science Logic (CSL 2021)}
\EventShortTitle{CSL 2021}
\EventAcronym{CSL}
\EventYear{2021}
\EventDate{January 25--28, 2021}
\EventLocation{Ljubljana, Slovenia (Virtual Conference)}
\EventLogo{}
\SeriesVolume{183}
\ArticleNo{34}

\begin{document}

\maketitle

\begin{abstract}
Contextuality is a key signature of quantum non-classicality, which has been shown to play a central role in enabling quantum advantage for a wide range of information-processing and computational tasks.
We study the logic of contextuality from a structural point of view, in the setting of partial Boolean algebras introduced by Kochen and Specker in their seminal work.
These contrast with traditional quantum logic \`a la Birkhoff and von Neumann
in that operations such as conjunction and disjunction are partial, only being defined in the domain where they are physically meaningful.

We study how this setting relates to current work on contextuality such as the sheaf-theoretic and graph-theoretic approaches.
We introduce a general free construction extending the commeasurability relation on a partial Boolean algebra, \ie the domain of definition of the binary logical operations.
This construction has a surprisingly broad range of uses.
We apply it in the study of a number of issues, including:
\begin{itemize}
\item establishing the connection between the abstract measurement scenarios studied in the contextuality literature and the setting of partial Boolean algebras;
\item
formulating various contextuality properties in this setting, including probabilistic contextuality as well as the strong, state-independent notion of contextuality given by Kochen--Specker paradoxes, which are logically contradictory statements validated by partial Boolean algebras, specifically those arising from quantum mechanics;
\item 
 investigating a Logical Exclusivity Principle, and its relation to the Probabilistic Exclusivity Principle widely studied in recent work on contextuality
 as a step towards closing in on the set of quantum-realisable correlations;
 
\item
developing some work towards a logical presentation of the Hilbert space tensor product, using logical exclusivity to capture some of its salient quantum features.
\end{itemize}
\end{abstract}

\section{Introduction}
Kochen and Specker's seminal work on quantum contextuality used the formalism of partial Boolean algebras \cite{kochen1975problem}. In contrast to quantum logic in the sense of Birkhoff and von Neumann \cite{birkhoff1936logic}, partial Boolean algebras only admit physically meaningful operations.
In the key example of $\PH$, the projectors on a Hilbert space $\HH$, the operation of conjunction, \ie product of projectors, becomes a partial one, only defined on \emph{commuting} projectors.

In more recent work \cite{kochen2015reconstruction}, Kochen 
developed
a large part of the foundations of quantum theory in terms of partial Boolean algebras.
Heunen and van den Berg \cite{van2012noncommutativity} showed that every partial Boolean algebra is the colimit of its (total) Boolean subalgebras. Thus the topos approach to quantum theory \cite{isham1998topos} can be seen as a refinement, in explicitly categorical language, of the partial Boolean algebra approach.
In this paper, we relate partial Boolean algebras to current frameworks for contextuality, in particular the sheaf-theoretic \cite{AbramskyBrandenburger} and graph-theoretic \cite{CSW} approaches.

A major role in our technical development is played by a general universal construction for partial Boolean algebras, which freely generates a new partial Boolean algebra from a given one and extra commeasurability constraints (\cref{ssec:freeextension}, \cref{commexthm}). This result is proved constructively, using an inductive presentation by generators and relations. It is used throughout the paper as it provides a flexible tool, subsuming a number of other constructions: some previously known, and some new.

We describe a construction of partial Boolean algebras from graphical measurement scenarios, \ie  scenarios whose measurement compatibility structure is given by a binary compatibility relation, or graph. Empirical models, \ie correlations satisfying the no-signalling or no-disturbance principle, on these scenarios correspond bijectively to probability valuations, or states, on the corresponding partial Boolean algebras (\cref{ssec:measurementscenarios2pBAs}).

We then turn our attention to contextuality properties. We discuss how probabilistic contextuality is formulated in the setting of partial Boolean algebras (\cref{ssec:probcontextuality}), and show that the strong, state-independent form of contextuality considered by Kochen and Specker can be neatly captured using the universal construction mentioned above (\cref{ssec:ksproperty}, \cref{KScolimthm}).

We also consider questions of quantum realisability,
\ie aiming to characterise the logical structure of partial Boolean algebras of projectors on a Hilbert space,
and probability models that admit a Hilbert space realisation.
This motivates us to propose a \emph{Logical Exclusivity Principle} (LEP), which is always satisfied by partial Boolean algebras of the form $\PH$ (\cref{ssec:lep}).
We use a variant of our universal construction to show that there is a reflection between partial Boolean algebras and those satisfying LEP (\cref{ssec:reflection}, \cref{thm:reflection}).
We relate this Logical Exclusivity Principle to Specker's Exclusivity Principle for probabilistic models \cite{cabello2012specker}.
We show that if a partial Boolean algebra satisfies LEP, then all its states satisfy the Probabilistic Exclusivity Principle (PEP) (\cref{ssec:lepvpep}, \cref{prop:lepimpliespep}). Moreover,
we show that a state on a partial Boolean algebra satisfies PEP if it extends to one on its logically exclusive reflection, \ie the freely generated partial Boolean algebra satisfying LEP  (\cref{ssec:lepvpep}, \cref{thm:statesAlex}).

In a similar vein,
we consider the extent to which the tensor product operation on Hilbert spaces can be ``tracked'' by a corresponding operation on partial Boolean algebras. We first consider the tensor product described in \cite{van2012noncommutativity,kochen2015reconstruction}, which can be put in generator and relations form using our free construction (\cref{ssec:firsttensor}).
It is easily seen that it fails to capture all the relations holding in the partial Boolean algebra of projectors on the Hilbert space tensor product.
We then show that there is a natural monoidal structure on partial Boolean algebras satisfying LEP (\cref{ssec:leptensor}). This contrasts with the situation for standard contextuality models satisfying Specker's Exclusivity Principle, which are not closed under tensor product.
Both tensor product constructions above work by freely extending commeasurability starting from the coproduct of partial Boolean algebras. We show that such an operation never gives rise to Kochen--Specker paradoxes (\cref{ssec:tensor_ks}). This can be seen as a limitative result for using such an approach
to fully capture the Hilbert space tensor product in logical form, in terms of partial Boolean algebras.

We conclude with a discussion of some natural questions that arise from our results (\cref{sec:discussion}).

\section{Partial Boolean algebras}

\subsection{Basic definitions}

A partial Boolean algebra $A$ is given by a set (also written $A$), a reflexive, symmetric binary relation $\comm$ on $A$, read as ``commeasurability'' or ``compatibility'',
constants $0$ and $1$, a total unary operation $\neg$, and partial binary operations $\meet$ and $\join$ with domain $\odot$.
These must satisfy the following property: every set $S$ of pairwise-commeasurable elements must be contained in a set $T$ of pairwise-commeasurable elements which forms a (total) Boolean algebra\footnotemark\ under the restrictions of the given operations.

\footnotetext{From now on, whenever we say just ``Boolean algebra'', we mean total Boolean algebra.}

Morphisms of partial Boolean algebras are maps preserving commeasurability, and the operations wherever defined. This gives a category $\pba$.

Heunen and van den Berg show in \cite[Theorem 4]{van2012noncommutativity} that
 every partial Boolean algebra is the colimit, in $\pba$, of the diagram $\CA$ consisting of its Boolean subalgebras and the inclusions between them.

\subsection{Colimits and free extensions of commeasurability}\label{ssec:freeextension}

In \cite{van2012noncommutativity}, the category $\pba$ is shown to be cocomplete. Coproducts have a simple direct description. The coproduct $A \oplus B$ of partial Boolean algebras $A$, $B$ is their disjoint union with $0_A$ identified with $0_B$, and $1_A$ identified with $1_B$. Other than these identifications,  no commeasurability holds between elements of $A$ and elements of $B$. By contrast, coequalisers, and general colimits, are shown to exist in \cite{van2012noncommutativity} by an appeal to the Adjoint Functor Theorem. One of our technical contributions is to give a direct construction of the needed colimits, 
by an inductive presentation.\footnote{For a well-known discussion of the advantages of an explicit construction over an appeal to the Adjoint Functor Theorem, see \cite[p. \textit{xvii}]{johnstone1997topos}.}

More generally, we use this approach to prove the following result, which freely generates from a given partial Boolean algebra a new one where prescribed additional commeasurability relations are enforced between its elements.

\begin{theorem}
\label{commexthm}
Given a partial Boolean algebra $A$ and a binary relation $\ocirc$ on $A$, there is a partial Boolean algebra $A[\ocirc]$ such that:
\begin{itemize}
\item there is a $\pba$-morphism $\fdec{\eta}{A}{A[\ocirc]}$ satisfying $a \ocirc b \IMP \eta(a) \odot_{A[\ocirc]} \eta(b)$;
\item for every partial Boolean algebra $B$ and $\pba$-morphism $\fdec{h}{A}{B}$ satisfying $a \ocirc b \IMP h(a) \odot_B h(b)$, there is a unique $\pba$-morphism $\fdec{\hat{h}}{A[\ocirc]}{B}$ such that
$h = \hat{h} \circ \eta$, \ie such that the following diagram commutes
\[ \begin{tikzcd}
A \ar[rd, "h"'] \ar[r, "\eta"] 
 & A[\ocirc] \ar[d, " \hat{h}"] \\
 & B
\end{tikzcd}
\]
\end{itemize}
\end{theorem}

We do not require that the relation $\ocirc$ include the commeasurability relation $\odot_A$ already defined on $A$. Of course, it is the case that
$A[\ocirc] \cong A[\odot_A \cup \ocirc]$ for any $\ocirc$, but it will be notationally convenient to allow an arbitrary relation $\ocirc$ in this construction. In particular, note that $A[\es] \cong A[\odot_A] \cong A$.

As already mentioned, this result is proved constructively, by giving proof rules for commeasurability and equivalence relations over a set of syntactic terms generated from $A$. In fact, we start with a set of ``pre-terms'' and also give rules for definedness.

We define the set of \emph{pre-terms} $P$ inductively, to be the closure of the set of generators
$G \, \defeq \, \setdef{\inl(a)}{a \in A}$
under the Boolean operations and constants.
The standard theory of Boolean algebras gives us an equational theory $\eqBool$  for terms in  the Boolean signature $\{ 0, 1, \wedge, \vee, \neg\}$ over variables $x, y, \ldots$
We have the usual notion of substitution of pre-terms for variables:
if $\vphi(\vec{x})$ is a Boolean term in the variables $x_1, \ldots , x_n$, and if $u_1, \ldots, u_n$ are pre-terms,
then $\vphi(\vec{u})$ is the pre-term which results from replacing $x_i$ by $u_i$ 
for all $i \in \enset{1, \ldots , n}$. 

We now define a predicate $\ex$ (definedness or existence), and binary relations $\comm$ and $\equiv$ on $P$, by the set of rules in \cref{tab:rules}.
To illustrate the first rule on the last line, consider the distributivity axiom: $x \wedge (y \vee z) \eqBool (x \wedge y) \vee (x \wedge z)$. Under the assumptions $t \odot u$, $u \odot v$, $t \odot v$, we can infer $t \wedge (u \vee v) \equiv (t \wedge u) \vee (t \wedge v)$. Note that in this rule  $\vphi(\vec{x})$ and $\psi(\vec{x})$ are pure Boolean terms, \ie they do not contain generators.

One can show the following by structural induction on derivations,
where $\vdash$ means derivability of an assertion from the rules.
\begin{lemma}
\label{prereqlemm}
For all pre-terms $t$ and $u$,
\begin{enumerate}
\item $\vdash t \odot u$ implies $\vdash t \ex$ and $\vdash u \ex$;
\item $\vdash t \equiv u$ implies $\vdash  t \ex$ and $\vdash u \ex$ and $\vdash t \odot u$.
\end{enumerate}
\end{lemma}

\begin{table}
\caption{Rules for free partial Boolean algebra with extended compatibility relation.}\label{tab:rules}
\begin{boxedminipage}{\textwidth}
\vspace{0.5\abovedisplayskip}\centering\noindent
$\begin{array}{c}
 \rinfer{a \in A}{\inl(a)\ex} \qquad \qquad \rinfer{a \odot_{A} b}{\inl(a) \odot \inl(b)} \qquad \qquad \rinfer{a \ocirc b}{\inl(a) \odot \inl(b)} 
 \\~\\
 \rinfer{}{0 \equiv \inl(0_A), \; 1 \equiv \inl(1_A)} \quad \rinfer{a \odot_{A} b}{\inl(a) \wedge \inl(b) \equiv \inl(a \wedge_{A} b), \; \inl(a) \vee \inl(b) \equiv \inl(a \vee_{A} b)} \quad \rinfer{}{\neg \inl(a) \equiv \inl(\neg_{A} a)}
 \\~\\
\rinfer{}{0 \ex, \; 1\ex} \qquad \qquad  \rinfer{t \odot u}{t \wedge u \ex, \; t \vee u \ex} \qquad \qquad \rinfer{t \ex}{\neg t \ex} 
\\~\\
\rinfer{t \ex}{t \odot t, \; t \odot 0, \; t \odot 1} \qquad \qquad \rinfer{t \odot u}{u \odot t} \qquad \qquad \rinfer{t \odot u, \; t \odot v, \; u \odot v}{t \wedge u \odot v, \; t \vee u \odot v}  \qquad \qquad \rinfer{t \odot u}{\neg t \odot u} 
\\~\\
\rinfer{t \ex}{t \equiv t} \qquad \qquad \rinfer{t \equiv u}{u \equiv t} \qquad \qquad \rinfer{t \equiv u,  \;  u \equiv v}{t \equiv v} \qquad \qquad \rinfer{t \equiv u, \; u \odot v}{t \odot v} 
\\~\\
\rinfer{\vphi(\vec{x}) \eqBool \psi(\vec{x}), \; \bigwedge_{i,j} u_i \odot u_j}{\vphi(\vec{u}) \equiv \psi(\vec{u})}  \qquad \qquad \rinfer{t \equiv t', \; u \equiv u', \; t \odot u}{t \wedge u \equiv t' \wedge u', \; t \vee u \equiv t' \vee u'}  \qquad \qquad \rinfer{t \equiv u}{\neg t \equiv \neg u}
\end{array}$
\vspace{0.5\belowdisplayskip}
\end{boxedminipage}
\end{table}

We define the set of terms $T \defeq \setdef{t \in P}{t\ex}$.
The relation $\equiv$ is an equivalence relation on $T$, by the rules in the fifth line.
We define a structure $A[\ocirc]$ as follows. The carrier is $T / {\equiv}$. The relation $\odot$ is defined by: $[t] \odot [u] \defIFF \; \vdash t \odot u$. This is well defined due to the last rule on the fifth line. The operations are defined by representatives: if $[t] \odot [u]$, then $[t] \wedge [u] \defeq [t \wedge u]$, etc. 
These are shown to be well defined using the congruence rules on the last line.
The first rule on the last line now suffices to show that any set of pairwise-commeasurable elements of $A[\ocirc]$
extends to a Boolean algebra, establishing the following proposition.

\begin{proposition}
 $A[\ocirc]$ is a partial Boolean algebra.
\end{proposition}

There is a map $\fdec{\eta}{A}{A[\ocirc]}$ sending $a$ to $[\inl(a)]$.
By the rules on the first two lines, this is a $\pba$-morphism which moreover satisfies $a \ocirc b \IMP \eta(a) \odot \eta(b)$.

Now, given a partial Boolean algebra $B$ and a morphism $\fdec{h}{A}{B}$ such that $a \ocirc b \IMP h(a) \odot_B h(b)$, we shall show that there is a unique partial Boolean algebra morphism $\fdec{\hat{h}}{A[\ocirc]}{B}$ such that $h = \hat{h} \circ \eta$.

We define a partial map $\gamma \colon P \pfn B$ by structural recursion on pre-terms:
\begin{align*}
\gamma(\inl(a)) &\defeq h(a) 
&
\gamma(t \meet u) &\defeq \gamma(t) \meet_{B} \gamma(u)
\\
\gamma(\neg t) &\defeq \neg_{B} \gamma(t)
&
\gamma(t \join u) &\defeq \gamma(t) \join_{B} \gamma(u) 
\end{align*}
Note that this map is partial because the operations $\meet_{B}$ and $\join_{B}$ are.

\begin{proposition}
For all pre-terms $t$ and $u$, the following conditions hold:
\begin{enumerate}
\item $\vdash t \ex$ implies $\gamma(t)$ is defined;
\item $\vdash t \odot u$ implies  $\gamma(t) \odot_{B} \gamma(u)$;
\item $\vdash t \equiv u$ implies $\gamma(t) = \gamma(u)$.
\end{enumerate}
\end{proposition}
\begin{proof}
The proof goes by structural induction on derivations from the rules.
It suffices to verify that each rule is a valid statement about the partial Boolean algebra $B$ when assertions about $t$, $u$ are replaced by the corresponding assertions about $\gamma(t)$, $\gamma(u)$.
Note that $\gamma(t) \odot_{B} \gamma(u)$ and $\gamma(t) = \gamma(u)$ are taken to imply, in particular, that $\gamma(t)$ and $\gamma(u)$ are well-defined elements of $B$.

For example, the third rule on the fifth line (transitivity of $\equiv$)
gets translated to
\[ \rinfer{\gamma(t) = \gamma(u), \; \gamma(u) = \gamma(v)}{\gamma(t) = \gamma(v)}  \]
which simply expresses transitivity of equality.
Most other cases are similar.

The first rule on the last line is the least straightforward.
The induction hypothesis gives $\gamma(u_i) \odot_{B} \gamma(u_j)$
for all $i$ and $j$, \ie $\enset{\gamma(u_1), \ldots, \gamma(u_n)}$ is a set of pairwise-commeasurable elements in $B$.
It can therefore be extended to a Boolean subalgebra of $B$.
This implies that for any pure Boolean term $\vphi(\vec{x})$,
$\gamma(\vphi(u_1, \ldots, u_n)) = \vphi^{B}(\gamma(u_1),\ldots,\gamma(u_n))$ is well defined in $B$,
and moreover that $\gamma(\vphi(\vec{u}))=\gamma(\psi(\vec{u}))$ whenever $\vphi(\vec{x}) \eqBool \psi(\vec{x})$, as required.
\end{proof}

\begin{proof}[Proof of Theorem 1]
We can now establish the required universal property. We define $\hat{h}([t]) \defeq \gamma(t)$. It follows straightforwardly from the definition of $\gamma$ together with the previous proposition that this is well defined and has the required properties.
\end{proof}

This result will prove to be very useful in what follows.

\subsubsection*{Coequalisers and colimits}

A variation of this construction is also useful, where instead of just forcing commeasurability, one forces equality. Given a partial Boolean algebra $A$ and a relation $\ocirc$ as before, we write $A[\ocirc,\equiv]$ for the algebra generated by the above inductive construction, with one additional rule:
\[ \rinfer{a \ocirc a'}{\inl(a) \equiv \inl(a')} \]
We can define a $\pba$-morphism $\fdec{\eta}{A}{A[\ocirc, \equiv]}$ by $\eta(a) \defeq [\inl(a)]$. Clearly this satisfies $a \ocirc a' \IMP \eta(a) = \eta(a')$.
A simple adaptation of the proof of \cref{commexthm} establishes the following universal property of this construction.
\begin{theorem}
Let $\fdec{h}{A}{B}$ be a $\pba$-morphism such that $a \ocirc a' \IMP h(a) = h(a')$. Then there is a unique $\pba$-morphism $\fdec{\hat{h}}{A[\ocirc, \equiv]}{B}$ such that $h = \hat{h} \circ \eta$.
\end{theorem}
This result can be used to give an explicit construction of coequalisers, and hence general colimits, in $\pba$. Given a diagram 
\[ \begin{tikzcd}
A \arrow[r, shift left, "f"] 
\arrow[r, shift right, "g"'] 
& B
\end{tikzcd}
\]
in $\pba$, we define a relation $\ocirc$ on $B$ by $b \ocirc b' \defeq \exists a \in A. \, f(a) = b \AND g(a) = b'$. Then, $\fdec{\eta}{B}{B[\ocirc,\equiv]}$ is the coequaliser of $f$ and $g$.

\subsection{States on partial Boolean algebras}

\begin{definition}\label{def:probvaluation}
A \emph{state} or \emph{probability valuation} on a partial Boolean algebra $A$ is a map $\fdec{\nu}{A}{[0,1]}$
such that:
\begin{enumerate}
\item\label{p:zero} $\nu(0) = 0$;
\item\label{p:complement} $\nu(\lnot x) = 1 - \nu(x)$;
\item\label{p:joinmeet} for all $x,y \in A$ with $x \odot y$, $\nu(x \join y) + \nu(x \meet y) = \nu(x) + \nu(y)$.
\end{enumerate}
\end{definition}

\begin{proposition}
A map $\fdec{\nu}{A}{[0,1]}$ is a state iff for every Boolean subalgebra $B$ of $A$, the restriction of $\nu$ to $B$ is a finitely additive probability measure on $B$.
\end{proposition}

\begin{lemma}\label{lemma:familyvaluation}
Let $A$ be a partial Boolean algebra.
There is a one-to-one correspondence between:
\begin{itemize}
\item states on $A$;
\item families $\family{\nu_S}_{S \in \CA}$ indexed by the Boolean subalgebras $S$ of $A$, where $\nu_S$ is a finitely additive probability measure on $S$ and $\nu_S = \nu_T \circ \iota_{S,T}$ whenever $S \subseteq T$.
\end{itemize}
\end{lemma}

\begin{lemma}\label{lemma:total:valuation-distribution}
Let $A$ be a finite Boolean algebra.
There is a one-to-one correspondence between states  on $A$ and probability distributions on its set of atoms.
\end{lemma}
\begin{proof}
Write $X$ for the set of atoms of $A$.
If $\fdec{\nu}{A}{[0,1]}$ is a state on $A$,
then 
\[
\sum_{x \in X}\nu(x) = \nu\left(\bigjoin X\right) \]
can be shown by induction on the size of $X$, using \cref{def:probvaluation}--\ref{p:zero} for the base case,
and using \cref{def:probvaluation}--\ref{p:joinmeet},\ref{p:zero} and the fact that $x \meet y = 0$ when $x$ and $y$ are distinct atoms for the induction step.
Since  $\bigjoin X = 1$, we conclude that  $\sum_{x \in X}\nu(x) = 1$ and so 
$\fdec{\nu|_X}{X}{[0,1]}$ is a probability distribution on $X$.

Conversely, if $\fdec{d}{X}{[0,1]}$ is a probability distribution, we extend it to the whole Boolean algebra
using the fact that any element is uniquely written
as the join of a set of atoms, as follows: for any $S \subseteq X$,
\[\nu\left(\bigjoin S\right) \defeq \sum_{x \in S} d(x) \Mdot\]
\end{proof}

\section{Graphical measurement scenarios and partial Boolean algebras}
\label{sec:measurementscenarios2pBAs}

\subsection{Measurement scenarios and (no-signalling) empirical models}

We consider the basic framework of the sheaf-theoretic approach introduced in \cite{AbramskyBrandenburger} to provide a unified perspective on non-locality and contextuality.
Our focus here will not be solely on the question of contextuality,
but also on principles that approximate the set of quantum-realisable behaviours.

Measurement scenarios provide an abstract notion of an experimental setup.
They model a situation
where there is a set of measurements, or queries, one can perform on a system,
but not all of which may be performed simultaneously.

In this paper, we focus on what we term `graphical' scenarios, where a subset of measurements is compatible (\ie can be performed together) if its elements are pairwise compatible. Hence, compatibility is specified simply by a binary relation.
A paradigmatic example is quantum theory, where compatibility is given by commutativity: a set of measurements (observables) can be performed together if and only if its elements commute pairwise.

Note that, in contrast to \cite{AbramskyBrandenburger}, we do not require that the set of measurements be finite. We do, however, consider only measurements with a finite set of outcomes.
This allows us to include within the scope of our discussion the scenario formed by all the quantum-mechanical observables on 
a system described by a finite-dimensional Hilbert space.

\begin{definition}
A  \emph{graphical measurement scenario} is a triple $\XGO$ consisting of:
\begin{itemize}
    \item a set $X$ of measurements,
    \item a reflexive, symmetric relation $\adj$ on $X$, indicating compatibility of measurements.
    \item a family $\family{O_x}_{x\in X}$ assigning a finite set $O_x$ of outcomes to each measurement $x \in X$.
\end{itemize}
A \emph{context} is a subset of measurements $\sigma \subseteq X$
that are pairwise compatible, \ie a clique of the relation $\adj$. We write $\Clique(\adj)$ for the set of contexts.
\end{definition}

A particular case of interest is that of measurement scenarios where every measurement is dichotomic, \ie has two possible outcomes.

Given a measurement scenario, an empirical model specifies particular probabilistic observable behaviour that may be displayed by a physical system.

\begin{definition}\label{def:empiricalmodel}
Let $\XGO$ be a measurement scenario.
A (no-signalling) \emph{empirical model} is a family
$\family{e_\sigma}_{\sigma \in \Clique(\adj)}$
where for each context $\sigma \in \Clique(\adj)$,
$e_\sigma$ is a probability distribution on the set
$\Ev(\sigma) \defeq \prod_{x\in\sigma}O_x$
of joint assignments of outcomes to the measurements in $\sigma$,
and such that $e_\sigma = e_\tau|_\sigma$ whenever $\sigma$ and $\tau$ are contexts with $\sigma \subseteq \tau$, where $e_\tau|_\sigma$ is marginalisation of distributions given as follows: for any $\vecs\in\Ev(\sigma)$,
\[e_\tau|_\sigma(\vecs) \defeq \sum_{\vect \in \Ev(\tau), \vect|_\sigma = \vecs}  e_\tau(\vect) \Mdot\]
Such an empirical model is said to be \emph{non-contextual} if there is a (global) probability distribution $d$ on the set $\Ev(X) = \prod_{x \in X}O_x$ 
that marginalises to the empirical probabilities,
\ie such that $d|_\sigma = e_\sigma$ for all contexts $\sigma \in \Clique(\adj)$.
\end{definition}

The marginalisation condition in the definition of empirical models ($e_\sigma = e_\tau|_\sigma$ for contexts $\sigma \subseteq \tau$) ensures that the probabilistic outcome of a compatible subset of measurements is independent of which other compatible measurements are performed alongside these. This is sometimes referred to as the \emph{no-disturbance condition} \cite{ramanathan2012generalized}, or \emph{no-signalling  condition} \cite{popescu1994quantum} in the special case of Bell scenarios.
This is a local compatibility condition, whereas non-contextuality can be seen as global compatibility: this justifies the slogan that contextuality arises from empirical data which is \emph{locally consistent but globally inconsistent} \cite{abrasmky2015ccp,abramsky2017contextuality}.

No-disturbance is satisfied by any empirical probabilities that can be realised in quantum mechanics \cite{AbramskyBrandenburger}.
However, this condition is much weaker than quantum realisability. Empirical models allow for behaviours that may be considered super-quantum, exemplified by the Popescu--Rohrlich (PR) box \cite{popescu1994quantum}.
A lot of effort has gone into trying to characterise the set of quantum behaviours by imposing some additional, physically motivated conditions on empirical models, leading to various approximations from above to this quantum set.

\subsection{Exclusivity principle on empirical models}

One candidate for a property that is distinctive for the quantum case has appeared in various formulations as Local Orthogonality \cite{fritz2013local}, Consistent Exclusivity \cite{henson2012quantum}, or Specker's Exclusivity Principle \cite{cabello2012specker}. We shall refer to it as the Probabilistic Exclusivity Principle (PEP), since it is expressed as a constraint on probability assignments.

Informally, it says that if we have a family of pairwise exclusive events, then their probabilities must sum to at most $1$. Of course, if all the events live on a single sample space, this would just be a basic property of probability measures. What gives the condition its force is that, in general, these events live on different, \emph{incompatible} contexts. Thus, it reaches beyond the usual view of contexts as different classical ``windows'' on a quantum system, in which incompatible contexts are regarded as incommensurable.

We can give a precise formulation of PEP in terms of empirical models as follows. First, we say that events $s \in \Ev(\sigma)$ and $t \in \Ev(\tau)$ are \emph{exclusive} if for some $x \in \sigma \cap \tau$, $s(x) \neq t(x)$.
The principle holds for an empirical model $\family{e_\sigma}_{\sigma \in \Clique(\adj)}$ if for any family $\{ s_i \in \Ev(\sigma_i) \}_{i \in I}$ of pairwise-exclusive events, then 
\[ \sum_{i \in I} e_{\sigma_i}(s_i) \, \leq \, 1. \]
This principle is valid in quantum-realisable empirical models, in which measurements correspond to observables, because incompatible (non-commuting) observables can share projectors, and exclusivity of outcomes with respect to common projectors implies \emph{orthogonality}.

Although we know that PEP does not fully characterise the quantum-realisable empirical models, it stands as an important and fruitful principle \cite{henson2012quantum, amaral2014exclusivity}.
We wish to study this principle from the perspective of partial Boolean algebras.

\subsection{From graphical measurement scenarios to partial Boolean algebras}
\label{ssec:measurementscenarios2pBAs}

To any graphical measurement scenario,
we can associate a partial Boolean algebra whose states correspond to empirical models.

\begin{definition}\label{def:AX}
Let $\XX = \XGO$ be a graphical measurement scenario. The partial Boolean algebra $\AX$ is defined as follows:
\begin{itemize}
\item For each measurement $x \in X$, take $B_x$ to be the finite Boolean algebra with atoms corresponding to the elements of $O_x$. We write $[x=o]$ for the atom of $B_x$ corresponding to the outcome $o \in O_x$.
\item Consider the partial Boolean algebra $A \defeq \bigoplus_{x \in X}B_x$, the coproduct of all the Boolean algebras $B_x$ taken in the category $\pba$. Note that all its elements are of the form $\inl_x(a)$ for a unique $x\in X$ and $a \in B_x$, except for the constants $0$ and $1$.
\item Define the following relation $\ocirc$ on the elements of $A$: 
\[\inl_x(a) \ocirc \inl_y(b) \Miff x\adj y \;\text{ or }\; a \in \enset{0,1} \;\text{ or }\; b \in \enset{0,1} \Mdot\]
\item Take $\AX \defeq A[\ocirc]$, the extension of $A$ by the relation $\ocirc$, as  given by \cref{commexthm}.
\end{itemize}
\end{definition}

We can give an alternative description using colimits.
\begin{definition}\label{def:cliquewise}
Let $\XX = \XGO$ be a graphical measurement scenario. The partial Boolean algebra $\BX$ is defined as follows:
\begin{itemize}
\item For each measurement $x \in X$, let $B_x$ be as in \cref{def:AX}.
\item For each context $\sigma \in \Clique(\adj)$,
let $B_\sigma \defeq \sum_{x \in \sigma} B_x$, the coproduct of all the $B_x$ with $x \in \sigma$, taken in the category $\ba$ of Boolean algebras.\footnote{Note that the set of atoms of such a coproduct Boolean algebra is the cartesian product of the sets of atoms of each of the summands. Hence an atom of $B_\sigma$ corresponds to an assignment of an outcome in $O_x$ to each measurement $x \in\sigma$.}
\item Given contexts $\sigma, \tau \in \Clique(\adj)$ with $\sigma \subseteq \tau$, there is a Boolean algebra homomorphism $\fdec{\iota^\tau_\sigma}{B_\sigma}{B_\tau}$ given by the obvious injection.
\item Take $\BX$ to be the colimit in the category $\pba$ of the diagram consisting of the Boolean algebras $\family{B_\sigma}_{\sigma \in \Clique(\adj)}$ and the inclusions $\family{\iota^\tau_\sigma}_{\sigma\subseteq\tau \in \Clique(\adj)}$.
\end{itemize}
\end{definition}

Note that the colimit in this instance can be given explicitly in a closed form,
as it is that of a diagram of Boolean algebras and inclusions satisfying the conditions of
Kalmbach's ``bundle lemma'' \cite[1.4.22]{kalmbach1983orthomodular}.
The carrier set of $\BX$ is the union of all the $B_\sigma$ modulo the identifications along inclusions $\iota^\tau_\sigma$. The Boolean subalgebras of $\BX$ are exactly those
in the family $\family{B_\sigma}_{\sigma \in \Clique(\adj)}$.

\begin{proposition}
The two descriptions coincide: for any  $\XX$, $\AX \cong \BX$.
\end{proposition}

\begin{proposition}\label{prop:partial:valuation-empiricalmodel}
For any graphical measurement scenario $\XX$, there is a one-to-one correspondence between
states on $\AX$
and empirical models on $\XX$.
\end{proposition}
\begin{proof}
This follows from the fact that the Boolean subalgebras of $\AX$ are the family $\family{B_\sigma}_{\sigma \in \Clique(\adj)}$, by applying \cref{lemma:familyvaluation,lemma:total:valuation-distribution}, 
and noting that
the condition  $\nu_S = \nu_T \circ \iota_{S,T}$ for $S \subseteq T$ on states on Boolean subalgebras
translates under the correspondence in \cref{lemma:total:valuation-distribution} to marginalisation of probability distributions.
\end{proof}

\section{Partial Boolean algebras and contextuality}

We consider some aspects of contextuality
formulated in the framework of partial Boolean algebras,
and relate them to the free construction from \cref{commexthm}.

\subsection{The Kochen--Specker property}
\label{ssec:ksproperty}

The Kochen--Specker theorem, as originally stated \cite{kochen1975problem}, is that there are partial Boolean algebras of Hilbert space projectors with no $\pba$-morphisms to $\Two$, the two-element Boolean algebra.
Since every (non-trivial\footnote{See the following discussion.}) Boolean algebra has a homomorphism to $\Two$, this implies that such a partial Boolean algebra $A$ has no morphism to \emph{any} (non-trivial) Boolean algebra.

Now, $\ba$ is a full subcategory of $\pba$. We know from \cite{van2012noncommutativity} that $A$ is the colimit in $\pba$ of the diagram $\CA$ consisting of its Boolean subalgebras and inclusions between them. Let $B$ be the colimit in $\ba$ of the same diagram $\CA$. Then, the cone from $\CA$ to $B$ is also a cone in $\pba$, hence there is a mediating $\pba$-morphism from $A$ to $B$.

To resolve the apparent contradiction, note that $\ba$ is an equational variety of algebras over $\mathbf{Set}$. 
As such, it is complete and cocomplete, but it also admits the one-element Boolean algebra $\One$, in which $0=1$. Note that the trivial Boolean algebra $\One$ does \emph{not} have a homomorphism to $\Two$.

We can conclude from the discussion above that a partial Boolean algebra satisfies the Kochen--Specker property of not having a morphism to $\Two$
if and only if the colimit in $\ba$ of its diagram of Boolean subalgebras is $\One$. 
In fact, we could formulate this property directly for diagrams of Boolean algebras, without referring to partial Boolean algebras at all: a diagram in $\ba$ is K--S if its colimit in $\ba$ is $\One$. We could say that such a diagram is ``implicitly contradictory'' since in trying to combine all the information in a colimit we obtain the manifestly contradictory $\One$.

Finally, this property admits a neat formulation in terms of the free extension of partial Boolean algebras by a relation,
reminiscent of the definition of a perfect group.

\begin{theorem}
\label{KScolimthm}
Let $A$ be a partial Boolean algebra. The following are equivalent:
\begin{enumerate}
    \item $A$ has the K--S property, \ie it has no morphism to $\Two$.
    \item The diagram $\CA$ of Boolean subalgebras of $A$ is K--S, \ie its colimit in $\ba$ is $\One$.
    \item $A[A^2] = \One$.
\end{enumerate}
\end{theorem}
\begin{proof}
The equivalence between the first two statements follows from the discussion above.
Now, all elements are commeasurable in $A[A^2]$, so it is a Boolean algebra. 
There is a morphism $A \longrightarrow \Two$
if and only if
there is a morphism $A[A^2] \longrightarrow \Two$,
by the universal property of $A[A^2]$ (in the $\Rightarrow$ direction) or composition with $\fdec{\eta}{A}{A[A^2]}$ (in the $\Leftarrow$ direction). Since $A[A^2]$ is a Boolean algebra, this is in turn equivalent to $A[A^2]$ being non-trivial. In other words, there is no morphism 
$A \longrightarrow \Two$
if and only if $A[A^2] = \One$.
\end{proof}

\subsection{Probabilistic contextuality}
\label{ssec:probcontextuality}

The notion of contextuality for states also
admits a formulation in this setting.

\begin{definition}\label{def:statecontextual}
A state $\fdec{\nu}{A}{[0,1]}$ on a partial Boolean algebra
$A$ is said to be \emph{non-contextual} if
it extends to $A[A^2]$, \ie if there is a state
$\fdec{\hat{\nu}}{A[A^2]}{[0,1]}$
such that $\nu = \hat{\nu}\circ\eta$.
\end{definition}

By the universal property of $A[A^2]$, this is equivalent to requiring that there be some Boolean algebra $B$,
a morphism $\fdec{h}{A}{B}$, and state $\fdec{\hat{\nu}}{B}{[0,1]}$ such that $\nu = \hat{\nu}\circ\eta$.

\begin{proposition}
Let $\XX$ be a graphical measurement scenario.
A state on $\AX$ is contextual  in the sense of \cref{def:statecontextual}
if and only if 
the corresponding empirical model under the correspondence of \cref{prop:partial:valuation-empiricalmodel} is contextual  in the sense of \cref{def:empiricalmodel}.
\end{proposition}

Note that if $A$ has the Kochen--Specker property, then $A[A^2] = \One$, and since there is no state on $\One$,
every state of $A$ is necessarily contextual.
An advantage of partial Boolean algebras is that the K--S property provides an intrinsic, logical approach to defining \emph{state-independent contextuality}.

\section{Exclusivity principles for partial Boolean algebras}

We now consider exclusivity principles from the partial Boolean algebra perspective.
This will subsume the previous discussion on PEP for empirical models in graphical measurement scenarios.

We introduce two exclusivity principles: one acts at the `logical' level, \ie the level of events or elements of a partial Boolean algebra, whereas the other acts at the `probabilistic' level, applying to states of a partial Boolean algebra.

\subsection{Logical exclusivity principle (LEP)}
\label{ssec:lep}

The basic ingredient is a notion of exclusivity between events (or elements) of a partial Boolean algebra.
Given a partial Boolean algebra $A$ and elements $a,b \in A$, we write
$a \leq b$ to mean that $a \odot b$ and $a \meet b = a$.
Note that the restriction of this relation $\leq$ to any Boolean subalgebra of $A$ coincides with the partial order underlying that Boolean algebra.

\begin{definition}
Let $A$ be a partial Boolean algebra.
Two elements $a, b \in A$ are said to be \emph{exclusive}, written $a \lex b$, if there is an element $c \in A$ such that
$a \leq c$
and
$b \leq \lnot c$.
\end{definition}

Note that $a \lex b$ is a weaker requirement than $a \meet b = 0$, although the two would be equivalent in a Boolean algebra. The point is that in a general partial Boolean algebra one might have exclusive events that are not commeasurable (and for which, therefore, the $\meet$~operation is not even defined).

\begin{definition}
A partial Boolean algebra is said to satisfy the
\emph{logical exclusivity principle (LEP)} if
any two elements that are  exclusive are also commeasurable, \ie if $\lex \; \subseteq \; \odot$.

We write $\epba$ for the full subcategory of $\pba$ whose objects are partial Boolean algebras satisfying LEP.
\end{definition}

\subsubsection*{Logical exclusivity and transitivity}
\label{leptranssec}

The logical exclusivity principle turns out to be equivalent to the following notion of transitivity \cite{lock1985connections,hughes1989structure}.

\begin{definition}
A partial Boolean algebra is said to be \emph{transitive} if for all elements $a, b, c$,
$a \leq b$ and $b \leq c$ implies $a \leq c$.
\end{definition}

Transitivity can fail in general for a partial Boolean algebra, since one need not have $a \odot c$ under the stated hypotheses.
Note that the relation $\leq$ on a partial Boolean algebra is always reflexive and anti-symmetric, so this condition is equivalent to $\leq$ being a partial order (globally) on $A$.
A partial Boolean algebra of the form $\Proj(\HH)$ is always transitive. 

\begin{proposition}
Let $A$ be a partial Boolean algebra. Then it satisfies LEP if and only if it is transitive.
\end{proposition}
\begin{proof}
Suppose that $A$ satisfies LEP, $a \leq b$, and $b \leq c$. Then $\neg c \leq \neg b$. Hence, by LEP, $a \odot \neg c$, and so $a \odot \neg \neg c = c$. Now, $a \meet c = (a \meet b) \meet c = a \meet (b \meet c) = a \meet b = a$, showing that $a \leq c$.

Conversely, suppose that $A$ is transitive, $a \leq c$, and $b \leq \neg c$. Then, $c = \neg \neg c \leq \neg b$, hence $a \leq \neg b$ by transitivity. In particular, $a \odot \neg b$, and so $a \odot \neg \neg b = b$.
\end{proof}

As an immediate consequence, any $\Proj(\HH)$ satisfies LEP.

It is shown in \cite{gudder1972partial} that a partial Boolean algebra is transitive if and only if it is an orthomodular poset.

\subsection{Probabilistic exclusivity principle (PEP)}

We now consider an analogous principle applying at the probabilistic level, \ie at the level on states of a partial Boolean algebra.

\begin{definition}
Let $A$ be a partial Boolean algebra.
A state $\fdec{\nu}{A}{[0,1]}$ on $A$ is said to satisfy the \emph{probabilistic exclusivity principle (PEP)} if for any set $S\subseteq A$ of pairwise-exclusive elements,
\ie such that $a \lex b$ for any distinct $a, b \in S$,
we have $\sum_{a \in S} \nu(a) \, \leq \, 1$.

A partial Boolean algebra is said to satisfy PEP if all of its states satisfy PEP.
\end{definition}

Note that the condition $\sum_{a \in S} \nu(a) \leq 1$ is true of any set $S$ of elements in a Boolean algebra satisfying $a \meet b = 0$ for distinct $a,b \in S$.

Note that this subsumes the discussion of the PEP at the level of empirical models.
If $\XX$ is a measurement scenario,
the correspondence in \cref{prop:partial:valuation-empiricalmodel} between empirical models on $\XX$ and states of $\AX$ restricts to a bijection between empirical models and states satisfying the probabilistic exclusivity principle.

\subsection{LEP \textit{vs} PEP}
\label{ssec:lepvpep}

The following result follows immediately from the definitions of partial Boolean algebras and states.

\begin{proposition}
\label{prop:lepimpliespep}
Let $A$ be a partial Boolean algebra satisfying the logical exclusivity principle. Then, any state on $A$ satisfies the probabilistic exclusivity principle.
\end{proposition}

In a general partial Boolean algebra $A$, however, not all states need satisfy the PEP.
A well-known example is the state on the partial Boolean algebra corresponding to a $(4,2,2)$ Bell scenario\footnote{This stands for a scenario in which there are $4$ parties, each of which can choose to perform one of $2$ measurements with $2$ possible outcomes.} which corresponds to two (independent) copies of the PR box \cite{fritz2013local}.

However, using the construction from \cref{commexthm}, we can construct from $A$ a new partial Boolean algebra, namely $A[\lex]$, whose
states yield 
states of $A$ that satisfy PEP.

\begin{theorem}\label{thm:statesAlex}
Let $A$ be a partial Boolean algebra. 
Then a state $\fdec{\nu}{A}{[0,1]}$ satisfies PEP if 
there is a state $\hat{\nu}$ of $\Alex$ such that 
\[ \begin{tikzcd}
A \ar[rd, "\nu"'] \ar[r, "\eta"] 
 & \Alex \ar[d, " \hat{\nu}"] \\
 & {[0,1]}
\end{tikzcd}
\]
commutes.
\end{theorem}
\begin{proof}
Let $\fdec{\nu}{A}{[0,1]}$ be a state,
and suppose it factorises through a state $\hat{\nu}$ of $\Alex$. Let $S \subseteq A$ be a set of pairwise exclusive events in $A$. Then $\setdef{\eta(a)}{a \in S}$ is a commeasurable subset of $\Alex$, hence it is contained in a Boolean subalgebra $B$ of $\Alex$. Since $\hat{\nu}$ must restrict to a finitely-additive probability measure on $B$, and since $\eta(a) \meet_{\Alex} \eta(b) = 0$ for all distinct $a, b \in S$, we have that 
\[\sum_{a \in S} \nu(a) = \sum_{a \in S}\hat{\nu}(\eta(a)) \leq 1 \Mdot\]
\end{proof}

\subsection{A reflective adjunction for logical exclusivity}
\label{ssec:reflection}

It is not clear whether the partial Boolean algebra $\Alex$ necessarily satisfies LEP. While the principle holds for all its elements in the image of $\fdec{\eta}{A}{\Alex}$,
it may fail to hold for other elements in  $\Alex$.

However, we can adapt the construction of \cref{commexthm}  to 
show that one can freely generate,
from any given partial Boolean algebra,
a new partial Boolean algebra satisfying LEP.
This LEP-isation is analogous to \eg the way
one can `abelianise' any group, or use Stone--\v{C}ech compactification to form a compact Hausdorff space from any topological space.

\begin{theorem}
\label{thm:reflection}
The category $\epba$ is a reflective subcategory of $\pba$, \ie the inclusion functor $\fdec{I}{\epba}{\pba}$ has a left adjoint $\fdec{\Exc}{\pba}{\epba}$.
Concretely, for any partial Boolean algebra $A$, there is a partial Boolean algebra $\Exc(A) = \Alexstar$
which satisfies LEP such that:
\begin{itemize}
    \item  there is a $\pba$-morphism $\fdec{\eta}{A}{\Alexstar}$;
   \item for any $\pba$-morphism $\fdec{h}{A}{B}$ where $B$ is a partial Boolean algebra $B$ satisfying LEP, there is a unique $\pba$-morphism $\fdec{\hat{h}}{\Alexstar}{B}$ such that $h = \hat{h}\circ\eta$, \ie such that the following diagram commutes:
\[ \begin{tikzcd}
A \ar[rd, "h"'] \ar[r, "\eta"] 
 & \Alexstar \ar[d, " \hat{h}"] \\
 & B
\end{tikzcd}
\]
\end{itemize}
\end{theorem}

The proof of this result follows from a simple adaptation of the proof of \cref{commexthm}, namely
adding the following rule to the inductive system presented in \cref{tab:rules}:
\[ \rinfer{u \meet t \equiv u, \; v \meet \neg t \equiv v}{u \odot v} \]

This rule will enforce the logical exclusivity principle,
and the universal property is proved in a manner similar to the proof of \cref{commexthm}.

\section{Tensor products of partial Boolean algebras}
\label{sec:tensor}

\subsection{A (first) tensor product by generators and relations}
\label{ssec:firsttensor}

In \cite{van2012noncommutativity}, it is shown that
$\pba$ has a monoidal structure, with $A \otimes B$ given by the colimit of the family
of Boolean algebras $C + D$, as $C$ ranges over Boolean subalgebras of $A$, $D$ ranges over Boolean subalgebras of $B$, and $+$ denotes the coproduct of Boolean algebras.

The tensor product in \cite{van2012noncommutativity} is not constructed explicitly: it relies on the existence of coequalisers in $\pba$, which is proved by an appeal to the Adjoint Functor Theorem.

Our \cref{commexthm}
allows us to give an explicit description of this construction using generators and relations.

\begin{proposition}
Let $A$ and $B$ be partial Boolean algebras.
Then 
\[A \otimes B \;\cong\; (A \oplus B)[\obar] \Mcomma\]
where $\obar$ is the relation on the carrier set of $A \oplus B$ 
given by $\inl(a) \obar \inr(b)$ for all $a \in A$ and $b \in B$.
\end{proposition}

This can be verified by comparing the universal property from \cref{commexthm} with   \cite[Proposition 30]{van2012noncommutativity}.

\subsection{A more expressive tensor product}
\label{ssec:leptensor}

There is a lax monoidal functor $\fdec{\Proj}{\Hilb}{\pba}$, which takes a Hilbert space to its projectors, viewed as constituting a partial Boolean algebra. The coherence morphisms $\PH \otimes \PK \to \PHoK$ are induced by the evident embeddings of $\PH$ and $\PK$ into $\PHoK$, given by $p \longmapsto p \otimes 1$, $q \longmapsto 1 \otimes q$.

It is easy to see that such morphisms are far from being isomorphisms. For example, if $\HH = \KK = \CC^2$, then there are (many) morphisms from $A = \Proj(\CC^2)$ to $\Two$, which lift to morphisms from $A \otimes A$ to $\Two$. However, by the Kochen--Specker theorem, there is no such morphism from $\Proj(\CC^4) = \Proj(\CC^2 \otimes \CC^2)$.

Interestingly, in \cite{kochen2015reconstruction} it is shown that the images of $\PH$ and $\PK$, for any finite-dimensional $\HH$ and $\KK$, generate $\PHoK$. This is used in \cite{kochen2015reconstruction} to justify the claim contradicted by the previous paragraph. 
The gap in the argument is that more relations hold in $\PHoK$ than in $\PH \otimes \PK$.
Nevertheless, this result is very suggestive.  In standard Boolean algebra theory, these images would satisfy the criteria for $\PHoK$ being the ``internal sum'' of $\PH$ and $\PK$ \cite{givant2008introduction}.
Evidently, for partial Boolean algebras, these criteria are no longer sufficient. This poses the challenge of finding stronger criteria, and a stronger notion of tensor product to match.

An important property satisfied by the rules in \cref{tab:rules} as applied in constructing $A \otimes B$ is that, if $t\ex$ can be derived, then $u\ex$ can be derived for every subterm $u$ of $t$.
This appears to be too strong a constraint to capture the full logic of the Hilbert space tensor product.

To see why this is an issue, consider projectors $p_1 \otimes p_2$ and $q_1 \otimes q_2$. To ensure in general that they commute, we need the conjunctive requirement that $p_1$ commutes with $q_1$ \emph{and} $p_2$ commutes with $q_2$. However, to show that they are \emph{orthogonal}, we have a disjunctive requirement: $p_1 \bot q_1$ \emph{or} $p_2 \bot q_2$. If we establish orthogonality in this way, we are entitled to conclude that  $p_1 \otimes p_2$ and $q_1 \otimes q_2$ are commeasurable, even though (say) $p_2$ and $q_2$ are not.
Indeed, the idea that propositions can be defined on quantum systems even though subexpressions are not is emphasised in \cite{kochen2015reconstruction}.

This leads us to define a stronger tensor product by forcing logical exclusivity to hold in the tensor product from \cite{van2012noncommutativity}. This amounts to composing with the reflection to $\epba$;
$\boxtimes \;\defeq\; \Exc \;\circ\; \otimes$.
Explicitly, we define the logical exclusivity tensor product by
\[ A \boxtimes B = (A \otimes B)[\lex]^*  = (A \oplus B)[\obar][\lex]^* . \]

This is sound for the Hilbert space model.  
More precisely, $\Proj$ is still a lax monoidal functor with respect to this tensor product. It remains to be seen how close it gets us to the full Hilbert space tensor product.

\subsection{Commeasurability extensions, Kochen--Specker, and Hilbert space tensor product}
\label{ssec:tensor_ks}

We can ask generally if extending commeasurability by some relation $R$ can induce the Kochen--Specker property in $A[R]$ when it did not hold in $A$. In fact, it is easily seen that this can never happen.
\begin{theorem}[K--S faithfulness of extensions]
\label{KSfaith}
Let $A$ be a partial Boolean algebra, and $R \subseteq A^2$ a relation on $A$. Then $A$ has the K--S property if and only if $A[R]$ does.
\end{theorem}
\begin{proof}
If $A$ does not have the K--S property, it has a morphism to a non-trivial Boolean algebra $B$. By the universal property of $A[R]$, there is a morphism $\fdec{\hat{h}}{A[R]}{B}$. Thus, $A[R]$ does not have the K--S property. Conversely, if there is a morphism $\fdec{k}{A[R]}{B}$ to a non-trivial Boolean algebra $B$, then $\fdec{k \circ \eta}{A}{B}$, so $A$ does not have the K--S~property.
\end{proof}

We can apply this in particular to the tensor product. 
\begin{corollary}\label{cor:tensorKS}
If $A$ and $B$ do not have the K--S property, then neither does $(A \otimes B)[\bot]^k$.
\end{corollary}
\begin{proof}
If $A$ and $B$ do not have the K--S property, they have morphisms to $\Two$, and hence so does $A \oplus B$.
Applying \cref{KSfaith} inductively $k+1$ times, one concludes that $(A \otimes B)[\bot]^k = (A \oplus B)[\obar][\bot]^k$ does not have the K--S property.
\end{proof}

Under the conjecture that $\Alexstar$ coincides with iterating $\Alex$ to a fixpoint, this would show that the logical exclusivity tensor product $A \boxtimes B$ never induces a Kochen--Specker paradox if none was already present in $A$ or $B$.

This can be seen as a limitative result, in the following sense. One of the key points at which non-classicality emerges in quantum theory is the passage from $\Proj(\CC^2)$, which does not have the K--S property, to $\Proj(\CC^4) = \Proj(\CC^2 \otimes \CC^2)$, which does.\footnote{Note that $\Proj(\CC^2) \cong \bigoplus_{i \in I} \Four_i$, where $I$ is a set of the power of the continuum, and each $\Four_i$ is the four-element Boolean algebra.}
By contrast, it would follow from \cref{cor:tensorKS} that $\Proj(\CC^2) \boxtimes \Proj(\CC^2)$ does not have the K--S property. Therefore, we need a stronger tensor product to track this emergent complexity in the quantum case.

\section{Discussion}
\label{sec:discussion}
A number of questions arise from the ideas developed in this paper.

\begin{itemize}
    \item First, we have shown that LEP implies PEP; that is, if a partial Boolean algebra satisfies Logical Exclusivity, then all its states satisfy Probabilistic Exclusivity. We conjecture that the converse holds.
    \begin{conjecture}
    \label{pepimpletconj}
    PEP $\IMP$ LEP.
    \end{conjecture}
    
    \item Similarly, we conjecture the converse to \cref{thm:statesAlex}.
    \begin{conjecture}
    If state $\nu$ of a partial Boolean algebra $A$ satisfies PEP, then there is a state $\hat{\nu}$ of $\Alex$ such that $\nu = \hat{\nu}\circ\eta$.
    \end{conjecture}
    This would amount to generalising the universality of $\Alex$ from $\pba$-morphisms to states. It would
yield a one-to-one correspondence between states
of $A$ satisfying PEP and states of $\Alex$.
    \item Proving the conjecture above would involve extending a state on a partial Boolean algebra $A$ to a state on $A[\ocirc]$.
    A similar operation was achieved for partial Boolean algebras arising from measurement scenarios  in \cref{prop:partial:valuation-empiricalmodel}, because in that case \cref{def:cliquewise} provided a simple description of the Boolean subalgebras of $A[\ocirc]$. Is an analogous description possible for the general case considered in \cref{commexthm}, or at least for the particular case of $\Alex$?
   \item A classic result by Greechie \cite{greechie1971orthomodular} constructs a class of orthomodular lattices which admit no states. Since orthomodular lattices are 
   transitive partial Boolean algebras (see \eg \cite{van2012noncommutativity}),
   this means that there are examples of partial Boolean algebras satisfying LEP which admit no states. Is there a partial Boolean algebra not satisfying LEP which admits no states? This would provide a counter-example to \cref{pepimpletconj}.
    \item There are some technical questions relating to the $\Alexstar$ construction:
    \begin{itemize}
        \item Is it a completion (\ie is the reflector a faithful functor)?
        \item Is it the same as iterating the $\Alex$ construction to a fixpoint?
        \item Is the relation of $\Alexstar$ to $\Alex$ an instance of a more general relationship between iterating an inductive construction, and adding a rule to the inductive construction itself?
    \end{itemize}
    \item Our discussion of tensor products led us to introduce a strong tensor product of partial Boolean algebras, $A \boxtimes B$. This brings us closer to an answer to the following particularly interesting question:
    \begin{question}
    Is there a monoidal structure $\circledast$ on the category $\pba$ such that the functor $\fdec{\Proj}{\Hilb}{\pba}$ is \emph{strong monoidal} with respect to this structure, \ie such that $\PH \circledast \PK \cong \PHoK?$
    \end{question}
    
    A positive answer to this question would offer a complete logical characterisation of the Hilbert space tensor product, and provide an important step towards giving logical foundations for quantum theory in a form useful for quantum information and computation.
    
    \item We recall the following quotation from Ernst Specker given in \cite{cabello2012specker}:
    \begin{quotation}
    \noindent
    Do you know what, according to me, is the fundamental theorem of quantum
mechanics? \ldots\ That is, if you have several questions and you can answer any two of them, then
you can also answer all three of them. This seems
to me very fundamental.
    \end{quotation}
    
    This refers to the \emph{binarity} of compatibility in quantum mechanics. A set of observables is compatible if they are pairwise so. This is built into the definition of partial Boolean algebras, and it is why we only considered graphical measurement scenarios in this paper. However, in the general theory of contextuality, as developed \eg in \cite{AbramskyBrandenburger}, more general forms of compatibility are considered, represented by simplicial complexes.
    The notion of partial Boolean algebras in a broader sense introduced in \cite{czelakowski1979broadersense} seems suitable to deal with this more general format.
    How much of the theory carries over?
    \item Partial Boolean algebras capture logical structure. We have seen how this logical structure can be used to enforce strong constraints on the probabilistic behaviour of states. This is somewhat analogous to the role of possibilistic empirical models in \cite{AbramskyBrandenburger}. Can we lift the concepts and results relating to possibilistic empirical models in \cite{AbramskyBrandenburger,abramsky2016possibilities,abramsky2013relationalhidden} to the level of partial Boolean algebras?
    
    \item There is much more to be said regarding contextuality in this setting. In current work in progress, we are considering the following topics:
    \begin{itemize}
        \item A hierarchy of logical contextuality properties generalising those studied in \cite{AbramskyBrandenburger}.
        \item A systematic treatment of ``Kochen--Specker paradoxes'', \ie contradictory statements which can be validated in partial Boolean algebras.
        \item Constructions that transform state-dependent to state-independent forms of contextuality.
    \end{itemize}
\end{itemize}

\bibliography{logicontext}

\end{document}